\documentclass[letterpaper,conference]{IEEEtran}

\usepackage{cite}
\usepackage{graphicx}
\usepackage{amsmath}
\usepackage{enumerate}
\usepackage{amsthm}
\usepackage{mathtools}
\usepackage[linesnumbered,lined,boxed,commentsnumbered,ruled]{algorithm2e}
\usepackage[noend]{algpseudocode}

\newtheorem{prop}{Proposition}
\newtheorem{thm}{Theorem}

\theoremstyle{definition}

\makeatletter
\def\BState{\State\hskip-\ALG@thistlm}
\makeatother
\usepackage{color}
\usepackage{amssymb}

\IEEEoverridecommandlockouts
\begin{document}
	
	\title{Distributed Relay Selection in Presence of Dynamic Obstacles in Millimeter Wave D2D Communication
		\thanks{Durgesh Singh and Sasthi C. Ghosh are with the Advanced Computing \&  Microelectronics Unit, Indian Statistical Institute, Kolkata 700108, India. Email: durgesh.ccet@gmail.com, sasthi@isical.ac.in.}
		\thanks{Arpan Chattopadhyay is with the Department of Electrical Engineering, Indian Institute of Technology Delhi. Email: arpanc@ee.iitd.ac.in.}
		\thanks{This work was supported by the faculty seed grant and professional development allowance (PDA) of IIT Delhi.}
		\vspace{-10mm}}
	
	\author{
		Durgesh Singh, Arpan Chattopadhyay \& Sasthi C. Ghosh
		\vspace{-10mm}
	}

	\maketitle

	\begin{abstract}		
		Millimeter wave (mmWave) device to device (D2D) communication is highly susceptible to obstacles due to severe penetration losses and requires almost a line of sight (LOS) communication path. D2D channel condition is local to devices/user equipments (UEs) and hence is \textit{not} directly  visible to the base station (BS). Thus quality of the D2D channel needs to be propagated to BS by UEs which may incur some delay. Hence the solution provided by BS to UEs using this gathered channel information might become less useful to establish communication due to moving obstacles. These types of obstacles might not be known in advance and hence may cause unpredictable fluctuations to the D2D channel quality. Hence we seek to learn the D2D channels using the finite horizon partially observable Markov decision process (POMDP) framework to model the uncertainty in such kind of network environments with dynamic obstacles. The objective is to  minimize delay when channel quality deteriorates, by making UEs choose locally the best possible decision between i) to continue on the current relay link on which communication is taking place or ii) to switch to another good relay by exploring other possible UEs in its locality. We derive an optimal threshold policy which tells the UE to take appropriate decision locally. Later, we give a simplified and easy to implement stationary threshold policy which counts the number of successive acknowledgement failures, based on which UE make appropriate decision locally.  Through extensive simulation, we demonstrate that our approach  outperforms recent algorithms.
	\end{abstract}
	
	\IEEEpeerreviewmaketitle
	
	\section {Introduction}
	Device to device (D2D)  communication in 5G may bypass the base station (BS) to make devices or user equipments (UE) directly communicate with one another. It helps in reducing outage and reuse resources and to meet the increasing bandwidth requirements of devices.  Generally, D2D communication is studied for short distance communication which makes millimeter wave (mmWave) as the suitable candidate for it \cite{8014297}. Although mmWave has very high available bandwidth, but it suffers from very high propagation losses, which may be compensated using directional beams in multi-input multi-output (MIMO) antennas. However, the penetration loss is also very severe for mmWaves for  most of the outdoor materials \cite{6655403,7010536}. Hence, it renders mmWave  unsuitable in presence of such obstacles which may completely block the mmWave signal.  Selecting relays to avoid obstacles have been studied in various works \cite{6932503,p11_7504422,p8_7510705,p14_7450161,8292574,8340227}. Most of these work carry out analysis on static obstacles, the problem of choosing relays becomes more challenging where the obstacles are also moving. 
	
	The D2D channel condition might deteriorate rapidly due to obstacles and especially due to moving obstacles. This in turn causes link breakage and hence  packet loss and delay. The BS cannot sense the quality of D2D channel directly and thus such information needs to be communicated by  UEs to the BS. Using this gathered information, the BS may suggest source UE to continue communication via another relay. However this might incur some delays and by the time UEs get global solution provided by the BS, it may become less useful for UEs to communicate due to possible blockage by some dynamic obstacle. There can be other parameters local to a UE (like battery, channel availability, perceived throughput etc.) which may further creates problem in implementing the global solution \cite{8014287}. For mmWave communication, capturing dynamic obstacles is challenging task whose information may not be available apriori to the local nodes/UEs. Radars can be used to sense the obstacles \cite{8457255,DBLP:journals/corr/abs-1907-08500}, but it may be too expensive to place radars for detecting the moving obstacles. To deal with the uncertainty caused by the dynamic obstacles, a learning based approached using partially observed Markov decision process (POMDP) \cite{7080987,7996366,dp_book} is an appropriate choice. We may use past information of D2D channel quality to learn about it. In fact the dynamic  obstacle's presence  is also captured indirectly while learning the channel quality. 
	
	In this paper, we are modeling our problem of relay selection at each UE locally as a finite horizon POMDP to capture the uncertainty caused in a D2D channel due to moving obstacles. The state of D2D channel is not observable at the current time instant. It can only be observed after taking the decision to transmit packets to a chosen link in form of acknowledgements (ACKs). Information about dynamic obstacles are not known at BS a priori and it can only be learned using the feedback from UEs after communication has been established.  Even the ACKs can get lost due to presence of dynamic obstacles. A  given UE transmitting the data packets may initially choose the relay suggested by the BS. However at later time instants, the channel quality of the suggested link might deteriorate and may cause huge packet loss and delay. We use conditional probability of D2D channel quality given the ACKs history as the sufficient statistics which is also called the belief probability of the given link. We then derive an optimal policy which maps the  belief to a set of actions. An action chosen can be either to continue on the current link or to stop and explore other possibly available relay links. Later, by exploiting the derived policy structure, we obtain a stationary policy  which tells the UE whether to continue transmitting along the chosen relay link in case of several successive ACK failure. This helps UE to  stop sending the packets on the current link (after some successive ACKs failure) to avoid packet loss and mitigate delay. This method is compared with other state of art solutions i) based on a recent work which selects relays based on maximum throughput \cite{p14_7450161} and ii) received signal strength (RSS) based approach. We show in simulation that our proposed method outperforms other approaches. 
	
	Our contributions in this paper are summarized as follows:
	\begin{enumerate}
		\item We consider the effects of dynamic obstacles on D2D mmWave links, which is a new and challenging topic.
		\item We formulate the problem of relay selection as a POMDP, and show that the optimal policy checks whether a certain belief probability exceeds a threshold. This is a non-trivial result that required proof of several interesting intermediate results.
		\item Our optimal policy can be implemented locally at each node, thereby facilitating distributed implementation.
		\item The threshold policy is further reduced to counting the number of successive ACK failures, which is simple and easy to implement.
	\end{enumerate}
	
	The rest of the paper is organized as follows. 	System model is described in  section \ref{system_model}. The POMDP formulation is provided in  section \ref{problem_formulation}. Optimal policy structure is derived in section \ref{derivation_policy}. Numerical results are provided in section \ref{experiment}, followed by the conclusions  in section \ref{conclusion}. All proofs are provided in the appendix.
	
	\section{SYSTEM MODEL} \label{system_model}
	
	We are considering the device-tier of 5G D2D architecture mentioned in \cite{tehrani2014device}, where devices can communicate among themselves with or without the help from BS. The service region is divided into various \textit{zones} or \textit{grids} as shown in figure \ref{fig_service_region} with one BS. Each zone may have many UEs and is assumed to have atleast one D2D device which is ready to take part in D2D communication as a relay or source/destination node. 
	We define \textit{sending zone} as that zone where at least one UE wants to transmit data to an UE of some other zone. If $i$ is the sending zone then it may form connection to a UE of another zone $j\in \mathbb{U}^i$ , where $\mathbb{U}^i$ is the \textit{viable} relay zones of the zone $i$ which is given by the BS. A viable relaying zone of zone $i$ is one which is nearer to the zone containing the destination UE and is in the communication range of the zone $i$. When the UE in zone $i$ forms a connection with another UE of zone $j$, then it is termed as link $j$.
	Link is formed between UEs of two zones when they are in communication range of each other and the received signal strength is sufficient for the required data rate.
	Each UE can communicate with one another on mmWave channels using directional antennas. The received signal strength ($Q_{ij}$) on zone $j$ from zone $i$  is modeled as \cite{p14_7450161}:
	\begin{equation}
		Q_{ij}=\mu_{ij} \cdot P_i \cdot G_t\cdot G_r \cdot PL_{ij}
	\end{equation}
	where, $\mu_{ij}$ is the shadowing random variable, $P_i$ is the transmit power of UE $i$, $G_t$ \& $G_r$ are transmit and receive beam-forming gains respectively. $PL_{ij}$ is the distance dependent path loss function.
	
	Time is discretized  as $(nN+l)\delta$ as shown in figure \ref{fig_time_slots}, where $n$ belongs to set of nonnegative integers, $l$ takes integer values in $[0,N-1]$, $\delta$ is the smaller discretized time slot when the UEs transmit packets locally. It is assumed that $\delta$ (for each $l\in[0,N-1]$) is large enough to send one packet of size $L$ bytes. Here, $N$ is the number of time slots (of $\delta$ duration) between two consecutive global decisions by the BS. Global decision by BS is made at time when $nN+l$ is divisible by $N$. At this time instant BS takes the channel state information from all UEs in the service region and gives the decision of best relaying UE of a given zone for a given source UE. Hence, in between two consecutive time instants when BS can make global decision, a UE can send at-most $N$ packets of size $L$ to another UE. Note that at time $l=0$, the UE chooses the relay link suggested by the BS and at time  $l\in \{1,2,\cdots,N-1\}$, BS has no control over the UEs.  At global time instants, BS sends two types of information to UEs, i) the best relay UE (or node) for a given source UE and ii) viable relaying zones $\mathbb{U}^i$ for given source zone $i$, hence the zone $i$ may choose an appropriate zone for relaying data from the set $\mathbb{U}^i$.
	
	There are static and dynamic obstacles in the service region. There is \textit{no} facility like radars (to track them) available at BS.  The behavior of dynamic obstacles  are not known a priori and need to be learned from the received acknowledgement of sent packets in an on-line fashion. Since mmWaves are highly susceptible to obstacles and suffer from severe penetration loss, we assume that even a single moving or static obstacle may break an already established D2D link and can cause packet loss. It is assumed that the mobility of UEs  in a zone $i$ for $N\delta$ duration do not bring them outside the zone and this do not cause link outage. Hence the only factors responsible for link breakage and packet loss are obstacles and channel condition due to fading.
	
	
	The source/relay node takes local decision when the current link quality is not good enough and the node locally explores and switches to another one-hop node by incurring penalty. This exploration is done for the zone's set $\mathbb{U}^i$ given by BS to find out the best relaying zone for that time instant. Note that the UE is using directional mmWave antennas for exploring the neighbors and this time is assumed to cause some significant delay with respect to the duration $\delta$. Here both exploring and packet loss is assumed to consume one time unit $\delta$. It is assumed that the relay link is established within this exploration time.
	
	\begin{figure}[h!]
		\centering
		\includegraphics[width=0.35\textwidth]{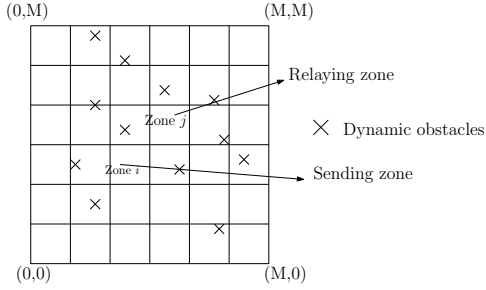}
		\caption{Service region divided into zones along with dynamic obstacles.}
		\label{fig_service_region}
	\end{figure}

	\begin{figure}[h!]
		\centering
		\includegraphics[width=0.35\textwidth]{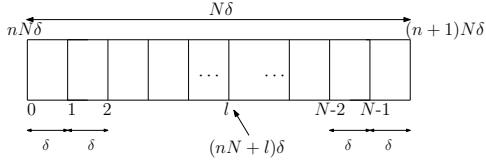}
		\caption{Discretized time slots with the smallest slot duration of $\delta$.}
		\label{fig_time_slots}
	\end{figure}
	
	\section{Problem Formulation as POMDP} \label{problem_formulation}
	
	Zone $i$ is the sending zone which contains at least one UE which needs to transmit data to an UE of some other zone $j\in \mathbb{U}^i$. This is termed as a link $j$ for the given sending zone  $i$. Hence zone $j$ may contain a relaying UE or the destination UE. Global decision for the best relay is given by the BS at the time instant $nN\delta$ to relay data packet till $(n+1)N\delta$ time instant. There are 
	both static and dynamic obstacles present in the environment which causes uncertainty in channel quality. Also, the  BS has no direct knowledge of the D2D channel conditions. This might deteriorate  the quality of relay link given by the BS. Which  might cause packet loss and delay in data transmission. We need to control this packet loss for $N\delta$ duration. However, BS do not have control over the data packets sent between time instants  $nN\delta$ and $(n+1)N\delta$. Hence the node locally needs to select for the best relay zone from $\mathbb{U}^i$ given the uncertainty of D2D channels and the current relay link has become bad. Since behavior  of channel condition is uncertain and unknown before actually establishing connection and transferring the packets, hence we will formulate this problem as a finite horizon POMDP \cite{dp_book}. 
	
	For the  duration between instants  $nN\delta$ and $(n+1)N\delta$, the time instants are referred as $l\in[0,N-1]$. Here we want to derive a decision criterion  to choose appropriate action (continue with the current relaying zone or explore and switch to some other zone) which lead the system to good state. Good state is defined by the minimum packet loss (in turn delay) considering all the required penalty costs. Hence our objective is to minimize the delay cost incurred due to packet loss while choosing appropriate relays and keeping the exploring and switching cost as low as possible. 
	For our POMDP problem, we will  describe state, action, observation, probabilistic structure of the problem, respective costs and cost function in upcoming paragraphs.
	
	For a given sending zone $i$, the state for all its possible relay  links $j\in U^i$ is written as $x_{l}^j \in\{0,1\}$. This signifies if relay link $j$ is in  good ($G$) or bad ($\overline{G}$) state for values $x_{l}^j =1$ and $x_{l}^j =0$ respectively. The relay link is in good state when the channel quality is as required and packet is transmitted successfully without getting blocked from obstacles, whereas in bad state the channel quality drops and hence packet loss occurs. The action set $\mathbb{A}$ is defined as \{explore \& switch to another link ($a_l^j=0$), transmit on current link (zone) ($a_l^j=1$)\}. The local node in zone $i$, makes observation at each smaller time instant $\delta$ after the packet is sent. This observation is in the form of ACK test which is denoted as $z_{l}^j\in\{0,1\}$. Here, $z_{l}^j=0$ represents that the acknowledgement is not received for link $j$ because link is bad which causes packet loss and similarly $z_{l}^j=1$ represents that the acknowledgement is  received and link is good and packet is transmitted successfully. We also represent $A$ and $\overline{A}$ as the ACK received or not ($z_l^j=1$ or $z_l^j=0$) respectively. Since ACK are quick and are available in negligible amount of time, for state $x_l^j=1$ and action $a_l^j=1$,  the observation (ACK) is $z_l^j$. The ACK may also be uncertain due to the unpredictable behavior of the given channel under consideration.

	The probabilistic structure of the observation assumed here is shown in figure \ref{fig_prob} and written as:
	\[P(z_{l}^j=1|x_{l}^j=1)=k; P(z_{l}^j=0|x_{l}^j=1)=1-k\]
	\[P(z_{l}^j=1|x_{l}^j=0)=0; P(z_{l}^j=0|x_{l}^j=0)=1\]
	If the system is in bad state with $x_{l}^j=0$ at time $l$, then the probability of obtaining good observation is zero ($P(z_{l}^j=1|x_{l}^j=0)=0$)  which is intuitive and obvious. The probabilistic structure assumed for the system state transition is given as:	
	\[P(x_{l+1}^j=1|x_{l}^j=1)=q; 	P(x_{l+1}^j=0|x_{l}^j=1)=1-q\]
	\[P(x_{l+1}^j=1|x_{l}^j=0)=s; 	P(x_{l+1}^j=0|x_{l}^j=0)=1-s\]
	Here $q$, $s$ and $k$ are respectively the probabilities that link is still good, bad link becomes good and the ACK is received successfully when the link is in good state. It is intuitive and legitimate to assume that $q>s$. The transition probability $1-q$ indicates that the good link becomes bad due to obstacles or signal fading. Similarly $(1-s)$ is the probability that bad link is still bad  (for  obstacles it indicates either obstacle is large in length or moving slowly and effecting the link for longer period).
	\begin{figure}[h!]
		\centering
		\includegraphics[width=0.30\textwidth]{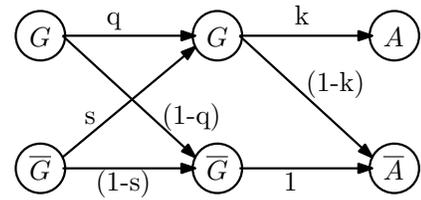}
		\caption{Probabilistic structure of the problem at a node locally.}
		\label{fig_prob}
	\end{figure}	
	For a given relaying zone $j$, let 	$I_l^j=(z_0^j,z_1^j,\cdots,z_l^j)$ denote the information vector available locally to the zone $i$ till smaller time instant $l$. Let us define $b_l$ as the conditional state distribution acting as the sufficient statistics or belief \cite{dp_book}(chapter 5) locally for the given relaying link $j$ as:
	\begin{equation}
		b_l^j=P(x_l^j=1|I_l^j)
	\end{equation}
	This equation signifies the probability that the relaying link is in good state given the previous history information. The estimator function of the local system is given as:
	\begin{equation}
		b_{l+1}^j=\Phi(b_l^j,z_{l+1}^j).
	\end{equation}
	Using Baye's rule we get,
	\begin{equation} \label{update_rule}
		b_{l+1}^j = \begin{cases}
			1, & \text{if } z_{l+1}^j=1\\
			\frac{[q b_l^j+ s  (1-b_l^j)] (1-k)}{[q b_l^j+s  (1-b_l^j)] (1-k) + [(1-q) b_l^j+ (1-s) (1-b_l^j)]} , & \text{if }z_{l+1}^j=0
		\end{cases}	
	\end{equation}
	The cost structure is defined as follows: when packet loss occurs then $C$ is the penalty (in terms of delay) incurred to overcome it. Since after a packet loss we may need to explore, hence, this is the same cost for exploration. When there is no packet loss then the cost incurred is $0$. Cost of testing for ACK is negligible and hence $0$. Here our objective is to derive a decision rule  to choose appropriate action (continue with the current relaying zone or explore and switch to some other zone) which leads the system to good state and in turn causes minimum packet loss  considering all the required costs. The expected cost is formulated as a dynamic program. At the end of the last period i.e., $(N-1)th$ period, the expected cost is defined as: 
	\begin{equation}
		J_{N-1}(b_{N-1}^j)=(1-kb_{N-1}^j) C. \label{eq_j1}
	\end{equation}
	Note that the for the last time instant $N-1$, packet loss can be due to two types of events: i) due to the link being in bad state and causing packet loss and  ii) when the link is in good state and the ACK is not received due to bad channel quality. For the time instant $l=N-2$, we have,
	\begin{equation}
		J_{N-2}(b_{N-2}^j)=\min\{C,\gamma_{N-2}+ E[J_{N-1}(b_{N-1}^j)] \label{eq_j2}
	\end{equation}
	where, $\gamma_{N-2}$ is the expected penalty paid due to packet loss at time $l=N-2$ which is $(1-k b_{N-2}^j) C$. The first term in minimization expression denotes the exploring \& switching cost and the second term denotes the cost for continuing in the current relay link. Similarly we can write the dynamic program for the general expression for each $l$ as:
	\begin{equation}
		J_{l}(b_{l}^j)=\min\{C,\gamma_l+E[J_{l+1}(\Phi(b_l^j,z_{l+1}^j))]\} \label{eq_j3}
	\end{equation}
	where, $\gamma_{l}$ is the expected penalty paid due to packet loss at time instant $l$ which is $(1-k b_{l}^j) C$. After solving this DP we will get a criterion, based on which the local decision can be made to switch the link or to remain on that link. Hence for a given relay zone $j$ at time instant $l$, we want to minimize the cost $J_l(b_l)$. The analysis of  this criterion is given in the next section where we derive a policy which maps the belief into action. The policy and hence the respective action taken optimize our objective function.
	
	\section{Derivation of the Optimal Policy} \label{derivation_policy}
	\subsection{Properties of $J_l(b)$}
	At the end of the $(N-1)th$ period, the expected cost is as mentioned in equation \eqref{eq_j1}. For the general expression for the time instant $l$ as mentioned in equation \eqref{eq_j3}, we can write it equivalently as:
	\begin{equation} 
		J_{l}(b_{l}^j)=\min\{C,A_l^j(b_l^j)\} \label{eq_dp}
	\end{equation}
	where, 
	\begin{multline}\label{eq_A_l}
		A_l^j(b_l^j)=\gamma_l+P(z_{l+1}^j=1|b_l^j) J_{l+1}(\Phi(b_l^j,1)) +\\ P(z_{l+1}^j=0|b_l^j)  J_{l+1}(\Phi(b_l^j,0))
	\end{multline}	
	For notation simplicity we will now remove the superscript $j$ from each of the respective notations, e.g., we will write $b_l^j$ as $b_l$ and $A_l^j()$ as $A_l()$. Hence $A_l^j(b_l^j)$ can now be denoted as $A_l(b_l)$. 
	
	We can reduce  $A_l(b_l)$ in equation \eqref{eq_A_l} to:
	\begin{multline} \label{equation_A_R}
		A_l(b_l)=(1-k b_{l}) C + (b_l q +(1-b_l) s)k  J_{l+1}(1)\\ + (1-(b_l q +(1-b_l) s)k) J_{l+1}(\Phi(b_l,0))
	\end{multline}
	
	As an example we will use this expansion to simplify equation \eqref{eq_j2} as:
	\begin{multline}
		J_{N-2}(b_{N-2}^j)=\min\{C,(1-k b_{N-2}^j) C+\\(1-(b_{N-2} q+(1-b_{N-2}) s) k) C\}.\label{eq_j2_expand}
	\end{multline}
	At the end of $(N-2)th$ period as shown in above  equation \eqref{eq_j2_expand}, the local node has calculated $b_{N-2}^j$ that the relay link $i$ is still the good node or not  and further decides whether to continue on the already selected relay link $j$ or needs to explore and switch to another relay node and incur extra cost $C$. In equation \eqref{eq_j2_expand}, 
	$(1-k b_{N-2}^j) C$ indicates the expected penalty incurred due to packet loss and $(1-(b_{N-2} q+(1-b_{N-2}) s) k) C$ indicates the expected cost to be incurred at the upcoming time instant $l=N-1$.
	
	We now show that  functions	$A_l(b_l)$ are piece-wise linear for each $l$ in proposition \ref{prop_1}. 	
	\begin{prop} \label{prop_1}
		$A_l(b_l)$ is piece-wise linear and concave in $b_l$ for each $l$.
	\end{prop} 	
	\begin{proof}
		See appendix.  
	\end{proof}
	\begin{prop} \label{prop_NEW}
		$\forall b_{l-1},b_l,b_{l+1} \in [0,1]$, 	\[A_{l-1}(b_{l-1}) \ge A_l(b_l) \ge A_{l+1}(b_{l+1})\]
		Also, $\forall b_l,b'_l$, $0 \le b_l < b'_l\le 1$, $A_l(b_l) \ge A_l(b'_l)$. 
	\end{prop}
	\begin{proof}
		See appendix.    
	\end{proof}

	\subsection{Policy Structure}
	The structure of an optimal policy for our POMDP problem is provided in the following theorem.
	\begin{thm}\label{thm_1}
		The optimal policy for our POMDP problem is a threshold policy. At any time instant $l \in \{0,1,\cdots,N-1\}$, the optimal action is to continue transmission on the current relay link if $b_l \geq \alpha_l$, and explore and switch to another better relay link if $b_l < \alpha_l$. Also, the threshold $\alpha_l \in [0,1]$ is non-increasing in $l$.
	\end{thm}
	\begin{proof}
		See appendix.  
	\end{proof}
	
	As  $l \rightarrow \infty$, $\alpha_l$ converges to some scalar $\overline{\alpha}$, since a decreasing sequence which is bounded below always converges. Hence, for very large horizon length $N$, the optimal policy can be approximated by a stationary threshold policy with a time-invariant threshold $\overline{\alpha}$. 
	
	Note that, if $z_l=1$, then $b_l=1$. Hence, without loss of generality, let us assume that $b_0=1$. If $z_0=0$, then $b_1=\Phi(b_0=1, z_0=0)=\frac{q-qk}{1-qk} < 1=b_0$. Now, it is easy to check that $\Phi(b,0)$ is a 
	strictly increasing function in $b$. Hence, $b_2=\Phi(b_1,0) < \Phi(b_0,0)=b_1$. Proceeding in this way, we can show that $b_l$ strictly decreases with $l$ whenever we observe several successive ACK failures.  We can define recursively a probability $\pi_m$ of getting $m$ successive ACK failure as: $\pi_1=\Phi(1,\overline{A})$,  $\pi_2=\Phi(\pi_1,\overline{A})$, $\cdots$.  Let $r$ be the smallest integer such that $\pi_r \le \overline{\alpha}$.  We can further simplify the stationary threshold policy as follows.
	
	{\bf Simplified stationary threshold policy: }
	{\em Let $r$ be the smallest integer such that $\pi_r \le \overline{\alpha}$. If there are $r$ successive ACK failures, explore and switch to another better relay link, else continue transmission on the current relay link.}
	
	%
	%

	\section{Simulation and Results} \label{experiment}	
	\subsection{Simulation Environment}
	We have divided the service region of $100~m \times 100~m$ square area into zones in form of grids each of dimension $10~m\times10~m$. Each zone have sufficient number of UEs which is enough to form a D2D link with UEs of other zones.  In the experiment,  $\delta$ is taken to be $100~ms$. Nodes are using directional transmitter and receiver antennas for $60~GHz$ frequency with $G_r=G_t=6~dB$ and we are considering a scenario where line of sight path loss exponent is $2.5$ and zero mean  log-normal shadowing random variable with standard deviation $3.5$ \cite{7974772,7109864}. Thermal noise density is $-174~dBm/Hz$ and devices are using $24~dBm$ transmit power.  Capacity of each link $(i,j)$ is $B\log_2(1+S_{ij})~bits/sec$, where $B=20~MHz$ \cite{al2014path} is bandwidth and $S_{ij}$ is the received signal to noise ratio. We are assuming  fixed packet  length of $65535~bytes$. There are maximum $16$ static and $D$ dynamic obstacles present in the environment, where $D\in\{0,16,32,48,64\}$. Static obstacles are placed uniformly in the service region. Each static obstacle is assumed to be of the dimension of a grid. Hence all communication going via that grid where there is an static obstacle will get blocked. Each dynamic obstacles is moving randomly and independently of each other and following a simple blockage model such that with probability $0.5$ it will block a given link otherwise it will not block the link.  We are assuming that a given zone $i$ can make connection with another zone out of given at-most $16$ neighboring zones surrounding it i.e $\mathbb{U}^i \le 16$.  We assume a single source-destination pair for simplicity and all other devices in a given zone may act as relay.
	
	We 	have written our own C++ custom code and run them on a \texttt{GNU} $4.8$ compiler on Intel core $i7$ machine. We run our experiments for around 10000 runs and take average results per run and per hop for the packet loss per packet delivered and end to end (E2E) delay per packet. Here packet loss per packet delivered is defined as the ratio of packet loss and  successfully delivered packets to the destination. E2E delay is the total time (in seconds) to send a packet successfully from source UE  to the destination UE ignoring the queuing delays. We are analyzing the results on these parameters with respect to number of dynamic obstacles $D$. We also analyzed the E2E delay on varying number of static obstacles. We are comparing the results of our proposed approach with metrics: 1) which selects relay link based on received signal strength (\texttt{RSS Based}) and 2) an approach which selects relay link based on maximum overall throughput (\texttt{ThroughPut Based})  \cite{p14_7450161}.
	
	\subsection{Simulation Results \& Analysis}	
	In figure \ref{fig-1}, we are comparing the results of packet loss per packet delivered successfully over the number of dynamic obstacles. We can see that as the number of dynamic obstacles is increased the packet loss  per packet delivered successfully is also increased. The reason is obvious due to the fact that as the number of dynamic obstacles increases, the chance of getting blocked also increases and hence the packet loss. Our proposed method outperforms other algorithms due to the fact that it learns the quality of the D2D links based on ACK and changes to another better relay when the quality of current D2D link deteriorates.
	
	In figure \ref{fig-2}, we are capturing the results of E2E delay per packet over the number of dynamic obstacles. Here also we can see that as the number of dynamic obstacles is increased the delay also increases. This is due to the fact that as the number of dynamic obstacles increases, packet loss increases and hence it causes extra delay. Our proposed method outperforms other algorithms due to the same reason as mentioned in above paragraph.
	
	Similarly, in figure \ref{fig-3}, we are capturing the results of E2E delay per packet over the number of static obstacles keeping no dynamic obstacles.  Here also we can see that as the number of static obstacles is increased the delay also increases since packet loss increases due to blockage from static obstacles too.
	
	\subsection{Discussions} 
	The proposed method can be run on each UE locally to choose an optimal relay at time instants when there is no control of the BS and the D2D channel quality becomes bad. It is evident from the results that as the number of obstacles increases, the packet loss increases rapidly. Since with higher number of the obstacles, the chance of a link to get blocked gets increased. Also the expected number of links getting blocked also increases.  However, it might be the case that the number of dynamic obstacles are so large that we may not find any D2D link which is free from the blockage due to obstacles. In this case our algorithm will not give any better links due to the reason that it will not find any link which satisfies the derived threshold policy $\overline{\alpha}$.  In such cases with very dense dynamic obstacles, empirically the packet loss is very negligible but the packet delivered successfully  is also very less and hence delay also might increase. In these scenarios, one appropriate solution would be to opt for the relays which are kept at some height above ground or to chose the transmission over traditional  micrometer waves of the BS which is less susceptible to the blockage by obstacles. 
	
	
	\begin{figure*}[h!]
		\centering
		\begin{minipage}[b]{.3\textwidth}
			\includegraphics[width=1.1\textwidth]{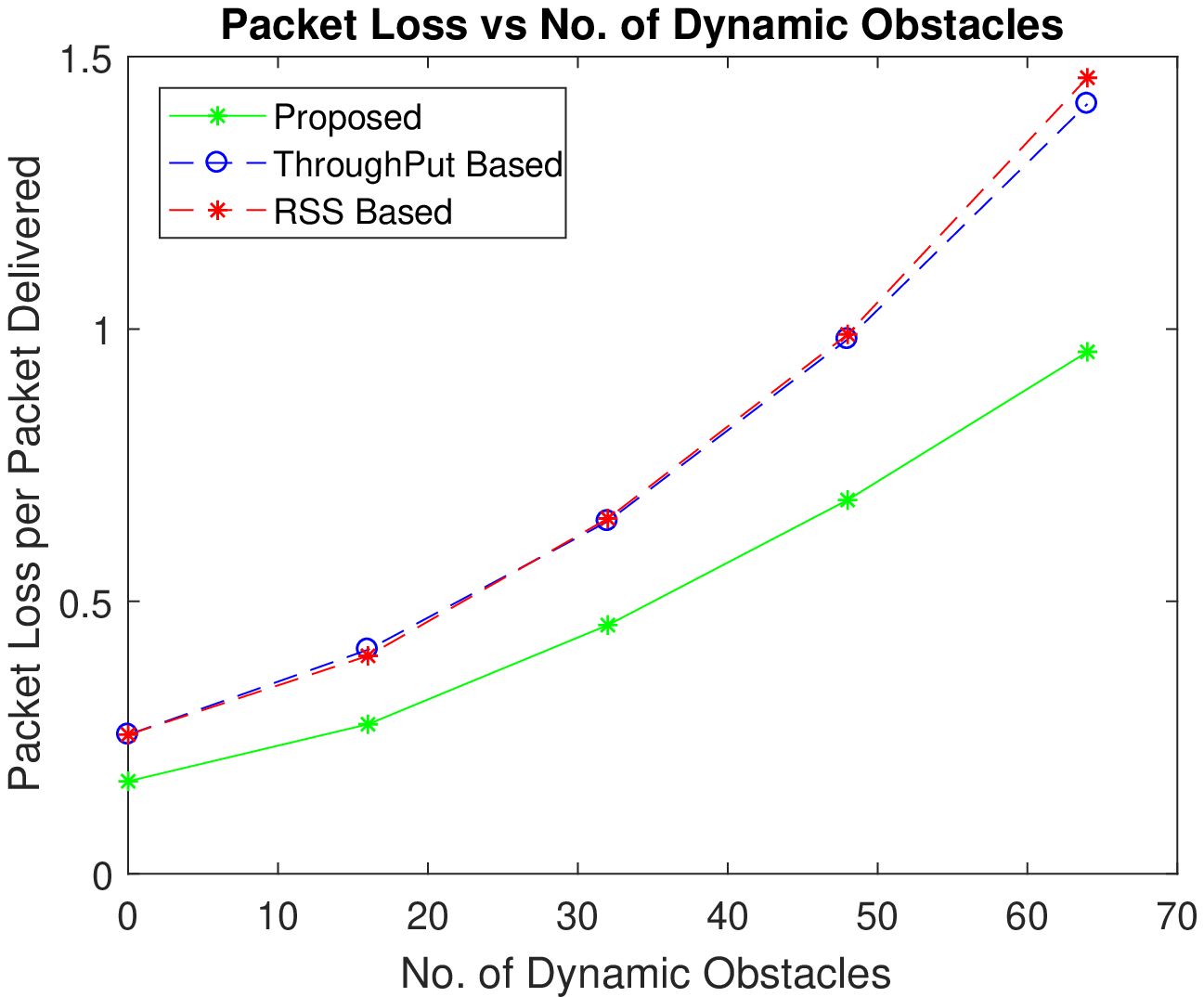}
			\caption{Packet loss per packet delivered vs No. of dynamic obstacles }
			\label{fig-1}
		\end{minipage}\qquad
		\begin{minipage}[b]{.3\textwidth}
			\includegraphics[width=1.1\textwidth]{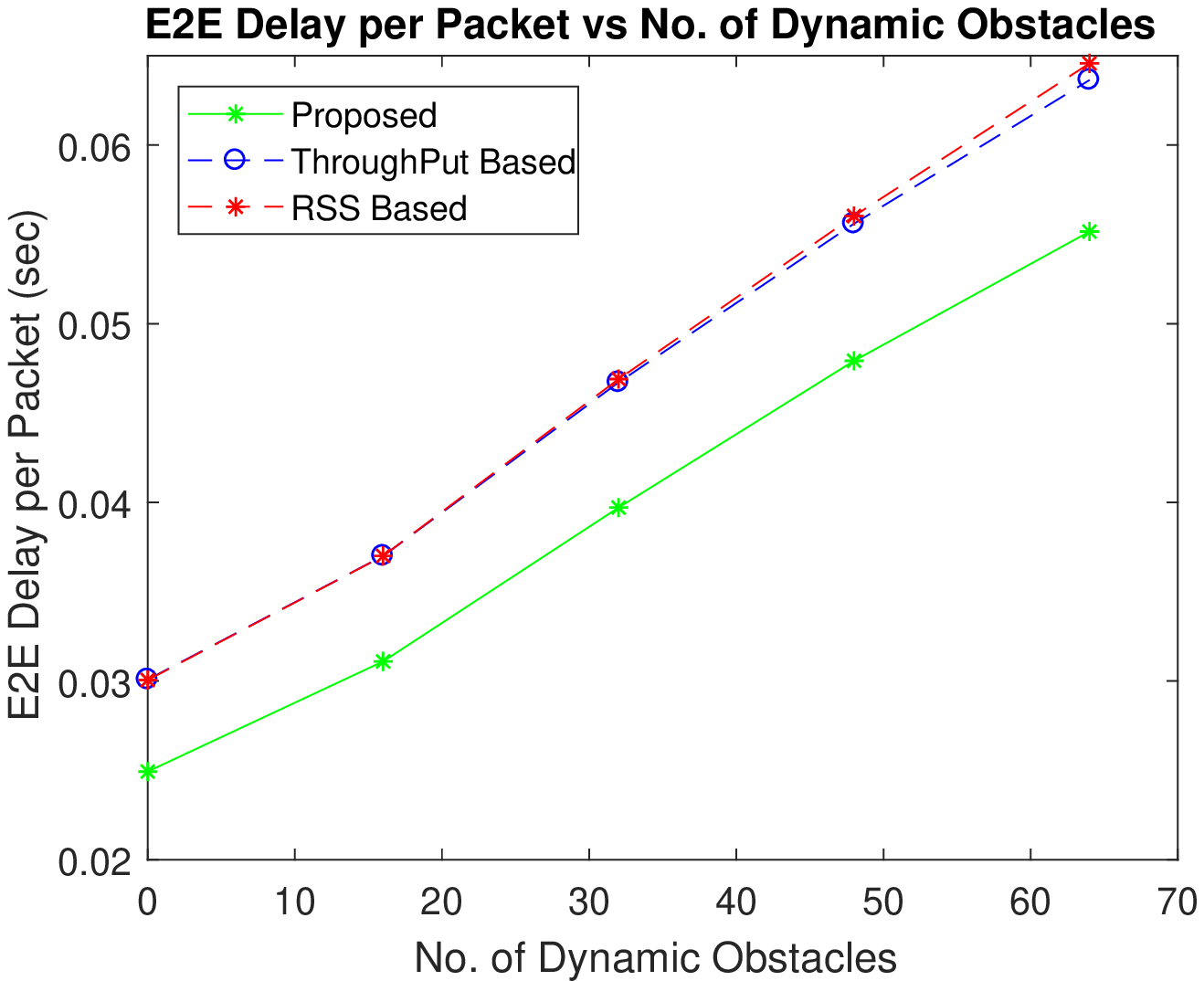}
			\caption{E2E delay per packet (in seconds) vs  No. of dynamic obstacles }
			\label{fig-2}
		\end{minipage}\qquad
		\begin{minipage}[b]{.3\textwidth}
			\includegraphics[width=1.1\textwidth]{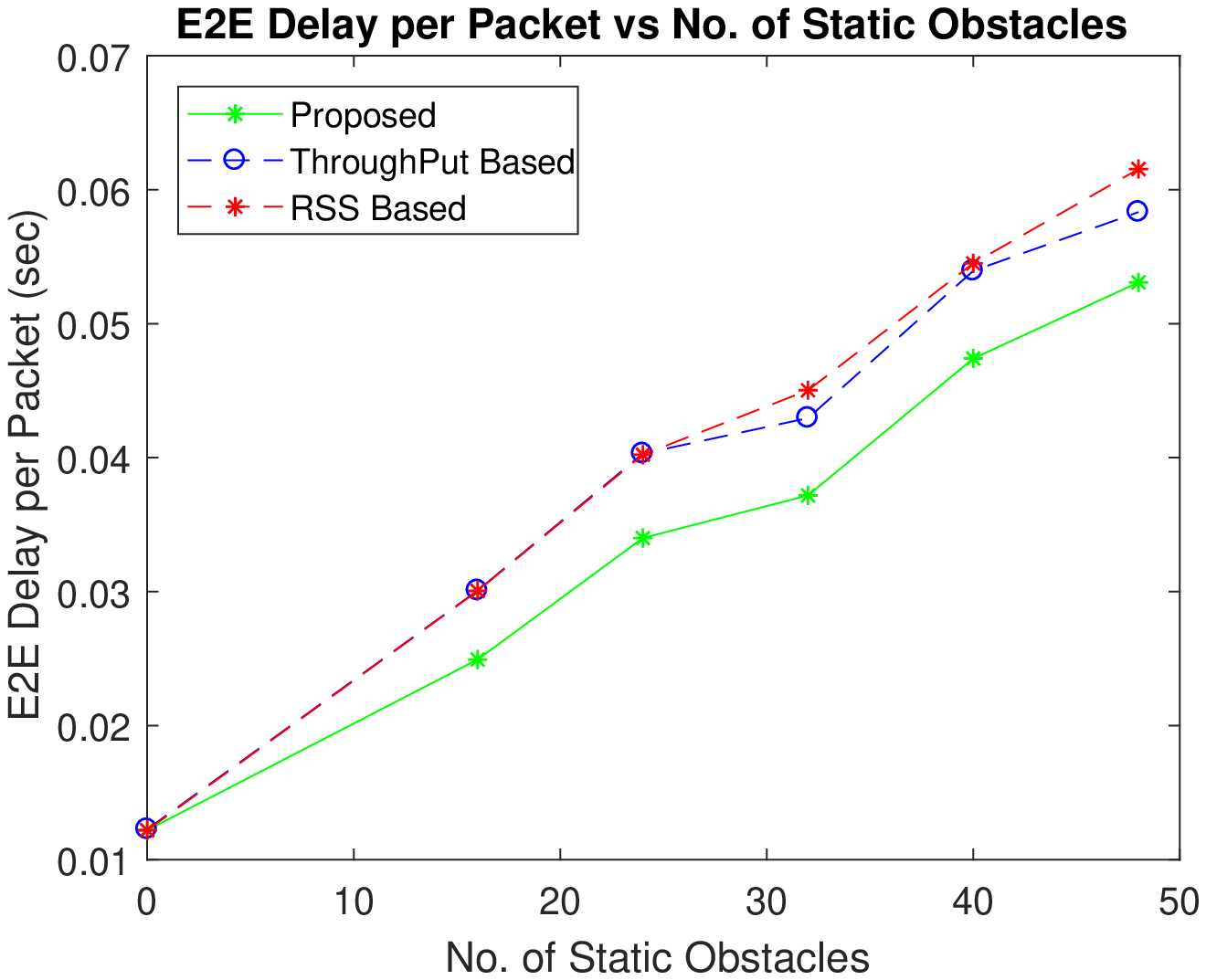}
			\caption{E2E delay per packet (in seconds) vs No. of static obstacles }
			\label{fig-3}
		\end{minipage}
	\end{figure*}

	\section{Conclusion}  \label{conclusion}
	Choosing a relay for D2D communication is a challenging task when the dynamic obstacles are present in the environment. This is because D2D channel quality is usually not directly visible to the BS. Dynamic obstacles cause unpredictable fluctuations to the D2D channel quality and hence they need to be learned from the channel statistics. We have modeled the problem of relay selection under the presence of dynamic obstacles as a finite horizon POMDP framework at each UE. This captures the uncertainty arising due to dynamic obstacles. Using this model, we have derived an optimal threshold policy for each UE that maps belief to action. We then derived a simple stationary policy which tells the UE to locally decide to either continue on the current relay link or to explore and switch to other relay link after successive ACK failures on the current relay link. This stationary policy  is simple and easy to implement. Through simulations we show that our approach captures the effects of dynamic obstacles and outperforms other state of art algorithms.

	\bibliography{ref} 
	\bibliographystyle{ieeetr}
	
	\onecolumn
	\appendix
	\noindent\textbf{Proof of proposition \ref{prop_1}:}

	We will prove this proposition for a general $b$ instead $b_l$. We will prove this by first showing that $J_l(b)$ is piece-wise linear and concave for each $l$ using induction. Then we prove our proposition.
	For time instant $(N-1)$, we have, $J_{N-1}(b)=(1-kb)C$, which is linear. 
	For time instant $(N-2)$, we have, $J_{N-2}(b)=\min\{C,(1-k b) C+ (1-(b q+(1-b) s) k) C\}$ which is also piece-wise linear and concave. 
	
	Assuming $J_{l+1}(b)$ is piece-wise linear and concave in $b$, we can say that for some suitable scalars, $\eta_1,\eta_2, \cdots, \eta_n$ and $\beta_1,\beta_2, \cdots, \beta_n$, $J_{l+1}(b)$ can be written as:
	\begin{equation}\label{eq_l+1_linear}
		J_{l+1}(b)=\min\{\eta_1+\beta_1b,\eta_2+\beta_2b,\cdots,\eta_n+\beta_nb\}.
	\end{equation}

	We can write $J_{l}(b)$=$\min\{C,A_l(b)\}$. Expanding it using equation \eqref{equation_A_R}, we get:
	\begin{multline}\label{eq_J_l}
		J_{l}(b)$=$\min\{C,(1-kb)C+(b q+(1-b) s) kJ_{l+1}(1)+(1-(b q+(1-b) s) k) J_{l+1}\bigg(\frac{(b q+ s (1-b)) (1-k)}{1-\{b q+s (1-b)\} k}\bigg)\}
	\end{multline}
	Substituting for $J_{l+1}$ from  equation \eqref{eq_l+1_linear} in above equation \eqref{eq_J_l}, we get:
	\begin{multline}\label{eq_J_l_expand}
		J_{l}(b)$=$\min\{C,(1-k b) C+(b q+(1-b) s) kJ_{l+1}(1)+(1-(bq+s(1-b))k) \min\{\eta_1+\beta_1\frac{(bq+ s(1-b))(1-k)}{1-\{bq+s(1-b)\}k},\\\eta_2+\beta_2\frac{(bq+ s(1-b))(1-k)}{1-\{bq+s(1-b)\}k},\cdots,\eta_n+\beta_n\frac{(bq+ s(1-b))(1-k)}{1-\{bq+s(1-b)\}k}\}\}
	\end{multline}
	We can further reduce above equation as:
	\begin{multline}\label{eq_J_l_linear}
		J_{l}(b)=\min\{C,(1-kb)C+(bq+(1-b)s) kJ_{l+1}(1)+\min\{\eta_1(1-\{bq+ s(1-b)\}k)+\beta_1(bq+ s(1-b))(1-k),\\\eta_2({1-\{bq+s(1-b)\}k})+\beta_2(bq+ s(1-b))(1-k),\cdots,\eta_n({1-\{bq+s(1-b)\}k})+\beta_n(bq+ s(1-b))(1-k)\}\}
	\end{multline}
	This is again piece-wise linear and concave in $b$. Thus the induction is complete. 
	
	Now we will show that $A_l(b)$ is also piece-wise linear and concave in $b$:
	\begin{multline} \label{eq_A_l_proof}
		A_l(b)=(1-k b) C+(b q +(1-b) s)k  J_{l+1}(1) + (1-(bq+(1-b)s)k)J_{l+1}(\Phi(b,0))
	\end{multline}
	The first term $(1-k b)C$ is linear in $b$. In the second term, $(b q +(1-b) s)k$ is linear in $b$ and $ J_{l+1}(1)$ is independent of $b$, hence overall $(1-k b)C+(bq +(1-b)s)k  J_{l+1}(1)$ is linear in $b$. Now we prove that $(1-(bq+(1-b)s)k)J_{l+1}(\Phi(b,0))$ is piece-wise linear in $b$ by expanding it using equation \eqref{eq_l+1_linear}, we get,
	\begin{multline}
		(1-(bq+(1-b)s)k)J_{l+1}(\Phi(b,0))=(1-(bq+(1-b)s)k)\min\{\eta_1+\beta_1\frac{(bq+ s(1-b))(1-k)}{1-\{bq+s(1-b)\}k},\eta_2+\\\beta_2\frac{(bq+ s(1-b))(1-k)}{1-\{bq+s(1-b)\}k},\cdots,\eta_n+\beta_n\frac{(bq+ s(1-b))(1-k)}{1-\{bq+s(1-b)\}k}\}
	\end{multline}
	We can reduce above to:
	\begin{multline}
		(1-(bq+(1-b)s)k)J_{l+1}(\Phi(b,0))=\min\{\eta_1(1-\{bq+ s(1-b)\}k)+\beta_1(bq+ s(1-b))(1-k),\\\eta_2(1-\{bq+ s(1-b)\}k)+\beta_2(bq+ s(1-b))(1-k),\cdots,\eta_n(1-\{bq+ s(1-b)\}k)+\beta_n(bq+ s(1-b))(1-k)\}.
	\end{multline}
	Since minimum of finite number of concave function is concave, $A_l(b)$ is piece-wise linear and concave in $b$ for all $l$.
	
	This proof is similar in spirit to an unsolved exercise given in \cite{dp_book}(chapter 5), however the DP and the estimator function of this paper are different from that given in the book. Hence we had to write a complete proof..

	\vspace{2em}
	\noindent\textbf{Proof of proposition \ref{prop_NEW}:}
	
	We will prove this proposition for a general $b$ for time instants $l-1$, $l$ and $l+1$ instead of $b_{l-1}$, $b_l$ and $b_{l+1}$. We will first prove  $J_l(b)\ge J_{l+1}(b)$, then we will use this to prove $A_l(b)\ge A_{l+1}(b)$. First we start for base case $l=N-1$ and the first term in recursion $l=N-2$: $J_{N-1}(b)=(1-kb)C$ and from equation \eqref{eq_j2_expand}, we have $J_{N-2}(b)=\min\{C,(1-kb)C+(1-(bq+(1-b)s)k)C\}$.
	We can easily see that $J_{N-2}(b)\ge J_{N-1}(b)$.
	We now prove it for first two terms of the recursion  $J_{N-2}(b)$ and $J_{N-3}(b)$. We can write $J_{N-3}(b)$ as:
	\begin{align}
		J_{N-3}(b)&= \min \{C, (1-kb)C + (bq+(1-b)s)kJ_{N-2}(1) + \notag 
		\quad (1-(bq+(1-b)s)k)J_{N-2}(\Phi(b,0))\} \\
		&\ge \min\{C, (1-kb)C + (bq+(1-b)s)kJ_{N-1}(1)  +  \notag 
		\quad (1-(bq+(1-b)s)k)J_{N-1}(\Phi(b,0))\} \\
		&=J_{N-2}(b)
	\end{align}
	Hence $J_{N-3}(b)\ge J_{N-2}(b)$. Similarly it proceeds for other $l$ and hence $J_{l}(b)\ge J_{l+1}(b)$. Now let us see this for $A_l(b)$ using previous proof for $J_{l}(b)$:
	\begin{align}
		A_{l}(b)&= (1-kb)C + (bq+(1-b)s)kJ_{l+1}(1) + (1-(bq+(1-b)s)k)J_{l+1}(\Phi(b,0))\\
		&\ge (1-kb)C + (bq+(1-b)s)kJ_{l+2}(1) + (1-(bq+(1-b)s)k)J_{l+2}(\Phi(b,0)) \\
		&=A_{l+1}(b)
	\end{align}
	Hence  $A_l(b) \ge A_{l+1}(b)$. This part is proved.
	
	To prove the second statement, We will prove it for general $b$ and $b'$ using induction. 
	
	We can see that $A_{N-2}(b)=(1-kb)C+(1-(bq+(1-b)s)k)C$ which is linear and non-increasing function in $b$. Let us assume this is true for $l+1$, such that $A_{l+1}(b)\ge A_{l+1}(b')$ for all $0\le b < b' \le 1$.  Let us now see for $l$:
	
	$A_l(b)=(1-k b) C + (b q +(1-b) s)k  J_{l+1}(1) + (1-(b q +(1-b) s)k) J_{l+1}(\Phi(b,0))$. We can rearrange the terms in this to write as:
	\begin{equation}
		A_l(b)= C+skJ_{l+1}(1)-kb(C-(q-s)J_{l+1}(1)) + (1-(b q +(1-b) s)k) J_{l+1}(\Phi(b,0))
	\end{equation}
	This equation can be further reduced as: 
	\begin{equation}
		A_l(b)= C+skJ_{l+1}(1)-kbC(1-(q-s)\min\{1,A_{l+1}(1)/C\}) + (1-(b q +(1-b) s)k)\min\{C,A_{l+1}(\Phi(b,0))\}
	\end{equation}
	
	In above equation, the term $\Phi(b,0) \le 1$ because it is a probability term. Function $\Phi(b,0)$ is increasing in $b$  for $q>s$ which is clear from equation \eqref{update_rule}. Hence using this and the induction hypothesis, we can say that $A_{l+1}(\Phi(b,0))$ is non-increasing and positive function in $b$. Also the term $(1-(b q +(1-b) s)k)$ is positive ($q>s$) and non-increasing in $b$. Hence we can say that the term $(1-(b q +(1-b) s)k)\min\{C,A_{l+1}(\Phi(b,0))$ is non-increasing in $b$. Now let us see the first term $C+skJ_{l+1}(1)-kbC(1-(q-s)\min\{1,A_{l+1}(1)/C\})$, here $C+skJ_{l+1}(1)$ is a constant and $kC(1-(q-s)\min\{1,A_{l+1}(1)/C\})$ is a positive quantity ($q>s$), hence we can say that $C+skJ_{l+1}(1)-kbC(1-(q-s)\min\{1,A_{l+1}(1)/C\})$ is non-increasing function in $b$. Sum of these two non-increasing functions is also a non-increasing function, hence, $A_l(b)$ is a non-increasing function in $b$.  Hence we can say that $A_l(b)\ge A_l(b')$ for all $0\le b < b' \le 1$.

	\vspace{2em}
	\noindent\textbf{Proof of theorem \ref{thm_1}:}
	
	First let us see the possible cases which exists for the given DP by fixing $l=N-2$ and at belief probabilities $0$ and $1$. For $b_{N-2}=0$, cost $C$ of exploration is always optimal and for $b_{N-2}=1$, if  $C<(1-k)C+(1-qk)C$ then continuing on that relay link $j$ for data transmission costs higher than exploring and switching on some other relay link. Similarly if  $C>(1-k)C+(1-qk)C$ then continuing on the  current link is the best option and exploring other links is never optimal. Hence we will see the following scenario where we can get the decision criterion for  choosing between exploring other links versus continuing on the same relay link. For this case choosing current link is optimal for $b_{N-2}=1$, if $C>(1-k)C+(1-q k)C$. If this condition is true, then there exists a scalar $\alpha_{N-2}$ with $0<\alpha_{N-2}<1$ that determines an optimal policy for the last period as: continue transmission on relay $j$ if $b_{N-2} \ge \alpha_{N-2}$ else stop transmission on relay $j$ and explore and switch to another better relay. 
	Since we are looking for the condition when the communication on given relay link can continue or not. When it cannot be continued then first action is chosen which stops the communication on the current relay link and exploration for new link begins.
	
	Using proposition \ref{prop_1} and proposition  \ref{prop_NEW}, we can say that the functions $y=C$ and $y=A_l(b_l)$ intersect at a single point and from the DP algorithm in equation \eqref{eq_dp}, we obtain that the optimal policy for each period is determined by the unique scalars $\alpha_l$ which are such that: $C=A_l(\alpha_l)$. Since we get a single point of intersection which decides the optimal choice for choosing an appropriate option, we can say that the optimal policy for the time period $l$ is given as:  continue transmission on relaying zone $j$ if $b_{l}\ge \alpha_l$, else stop transmission on relaying zone $j$ and explore and switch to another better relay. 
	
	Second part of this theorem: Using proposition \ref{prop_NEW}, $A_l(b_l)$ are monotonically non-increasing with respect to $l$. Hence we can say that sequence of $\alpha_l$  is also non-increasing with $l$ (using proposition \ref{prop_1} and \ref{prop_NEW}). 
	
\end{document}